\documentclass[12pt]{amsart}
\usepackage[foot]{amsaddr}
\usepackage{amsmath}
\usepackage{amsfonts}
\usepackage{amssymb}
\usepackage{graphicx}                                                                                               
\usepackage{subcaption}    
\usepackage[top=1.5in,bottom=1.5in,left=1.5in,right=1.5in]{geometry}
\usepackage[shortlabels]{enumitem}

\usepackage{bbm}
\newtheorem{thm}{Theorem}

\newtheorem{lemma}[thm]{Lemma}
\newtheorem{prop}[thm]{Proposition}

\newtheorem{defn}[thm]{Definition}

\newcommand{\ket}[1]{|#1\rangle}
\newcommand{\bra}[1]{\langle#1|}
\newcommand{\HH}{\mathcal{H}}
\newcommand{\NN}{\mathbb{N}}
\newcommand{\mcA}{\mathcal{A}}
\newcommand{\mcB}{\mathcal{B}}
\newcommand{\ZZ}{\mathbb{Z}}
\newcommand{\CC}{\mathbb{C}}
\newcommand{\gi}{e} 
\newcommand{\I}{\mathbbm{1}} 

\begin{document}

\title[Commuting-operator strategies for linear system games]{Perfect commuting-operator strategies \\ for linear system games}

\author{Richard Cleve$^*$}
\address{${}^* {}^\dagger {}^\ddagger$Institute for Quantum Computing, University of Waterloo, Canada}
\address{${}^* {}^\dagger$School of Computer Science, University of Waterloo, Canada}
\address{${}^*$ Canadian Institute for Advanced Research}

\author{Li Liu$^\dagger$}

\author{William Slofstra{$^\ddagger$}}
\email{\{cleve,l47liu,weslofst\}@uwaterloo.ca}


\begin{abstract}
    Linear system games are a generalization of Mermin's magic square game
    introduced by Cleve and Mittal. They show that perfect strategies for
    linear system games in the tensor-product model of entanglement correspond
    to finite-dimensional operator solutions of a certain set of
    non-commutative equations. We investigate linear system games in the
    \emph{commuting-operator} model of entanglement, where Alice and Bob's
    measurement operators act on a joint Hilbert space, and Alice's operators
    must commute with Bob's operators.  We show that perfect strategies in this
    model correspond to possibly-infinite-dimensional operator solutions of the
    non-commutative equations. The proof is based around a
    finitely-presented group associated to the linear system which arises from the non-commutative equations.
\end{abstract}

\maketitle

\section{Introduction}

Mermin~\cite{Mermin1990} implicitly considers a non-local game that is
sometimes called the \textit{magic square game} (see
also~\cite{Peres1990,Mermin1993,Aravand2004,CleveH+2004}). This game is based
around a system of linear equations over $\ZZ_2$ with nine variables and six
equations. In the game, Alice receives as input one of the six equations, and
Bob receives as input one of the variables from the same equation. Without
communicating with each other, Alice must output an assignment of the variables
in her equation, and Bob must output an assignment of his variable. The players
\textit{win} if and only if Alice's assignment satisfies her equation and their
assignments are consistent in the common variable. Remarkably, Alice and Bob
can always win Mermin's game if they use entanglement; there is no way to
achieve this without entanglement. 

Cleve and Mittal~\cite{CleveM2012} investigate 
the general case of games based
on binary linear systems%
\footnote{In fact, they consider a more general
scenario called \textit{binary constraint system games}, where each equation can
be based on an arbitrary boolean function of inputs.} of the form $Mx = b$,
where $M \in \ZZ_2^{m \times n}$ and $b \in \ZZ_2^m$. A solution of such a
system is a vector $x \in \ZZ_2^n$ such that $Mx=b$. It is 
convenient to
write these equations in multiplicative form, so a vector $x \in \{\pm 1\}^n$
satisfies equation $\ell$ if and only if 
\begin{equation*}
    x_1^{M_{\ell,1}} x_2^{M_{\ell,2}} \cdots x_n^{M_{\ell,n}} = (-1)^{b_{\ell}}.
\end{equation*} 
An equivalent way of writing equation $\ell$ is 
\begin{equation*}
    x_{k_1} x_{k_2} \dots x_{k_r} = (-1)^{b_{\ell}},
\end{equation*} 
where $V_{\ell} = \{k_1,k_2,\dots,k_r\} = \{ 1 \leq k \leq n : M_{\ell,k} =
1\}$ is the set of indices of variables in equation $\ell$. The non-local game
associated with a binary linear system $Mx=b$ is defined similarly to that of
the magic square game. A \textit{classical strategy} is one where Alice and Bob
do not share entanglement.  It can be shown that $Mx=b$ has a perfect classical
strategy (i.e., a strategy with success probability 1) if and only if the
system of equations has a solution.

An entangled quantum strategy is a strategy in which Alice and Bob share an
entangled quantum state $\ket{\psi}$. In the tensor-product model, $\ket{\psi}$
is a bipartite state in a tensor product $\HH_A \otimes \HH_B$, and Alice and
Bob's measurements of this state are modeled as observables on $\HH_A$ and
$\HH_B$ respectively. It is shown in~\cite{CleveM2012} that a binary linear system game
has a perfect entangled strategy in the tensor-product model if and only if the
linear system has a finite-dimensional operator solution in the following sense:

\begin{defn}[Operator solution of binary linear system]\label{def:operator-solution}
    An {\em operator solution} to a binary linear system $Mx=b$ is a sequence
    of bounded self-adjoint operators $A_1, \dots, A_n$ on a Hilbert space $\HH$ such that:
    \begin{enumerate}[(a)]
        \item $A_i^2 = \I$ (that is, $A_i$ is a binary observable) for all $1 \leq i \leq n$.
        \item If $x_i$ and $x_j$ appear in the same equation (i.e., $i, j \in V_{\ell}$
            for some $1 \leq \ell \leq m$) then $A_i$ and $A_j$ commute (we call this
            {\em local compatibility}).
        \item For each equation of the form $x_{k_1} x_{k_2} \dots x_{k_r} = (-1)^{b_l}$,
            the observables satisfy 
            \begin{equation*}
                A_{k_1} A_{k_2} \cdots A_{k_r} = (-1)^{b_{\ell}}\I
            \end{equation*} 
            (we call this {\em constraint satisfaction}).
    \end{enumerate}
    A {\em finite dimensional operator solution} to a binary linear system $Mx=b$ is
    an operator solution in which the Hilbert space $\HH$ is finite dimensional.
\end{defn}
The term local compatibility comes from quantum mechanics, where two
observables commute if and only if they are compatible in the sense that they
represent quantities which can be measured (or known) simultaneously. It is
noteworthy that the result of \cite{CleveM2012} applies even when the Hilbert
spaces $\HH_A$ and $\HH_B$ are allowed to be infinite dimensional; in this
case, the operator solutions will still be finite dimensional.

In this paper we are interested in the \emph{commuting operator model} for
entanglement, in which $\ket{\psi}$ belongs to a joint Hilbert space $\HH$, and
Alice and Bob's measurements are modeled as observables on $\HH$ with the
property that Alice's observables commute with Bob's observables. This
model---which clearly subsumes the tensor-product model---is used in algebraic
quantum field theory. For any non-local game, a finite-dimensional strategy in
the commuting-operator model can be converted into a strategy in the tensor
product model, but the precise relationship between the tensor-product model
and the commuting-operator model is unknown in general. We refer to
\cite{Tsirelson1993, scholz2008,junge2011,Fritz2012} for more discussion.

The main result of our paper is that a binary linear system game has a perfect
entangled strategy in the commuting operator model if and only the linear
system has a (possibly-infinite-dimensional) operator solution. As is typical
with results of this type (compare for instance \cite[Proposition
5.11]{PaulsenSSTW14}), the main difficulty arises in showing that an operator
solution can be turned into a perfect strategy. In particular, an operator
solution does not come with an entangled state. For this part of the proof, we
make use of the fact that the relations for operator solutions in Definition
\ref{def:operator-solution} resemble (aside from the appearance of the scalar $(-1)$) the
relations of a group presentation. If we represent $(-1)$ by a new
variable $J$, we get a finitely-presented group $\Gamma$, which we call the
\emph{solution group}. We can then construct a tracial state on the group algebra 
of $\Gamma$ to use as our entangled state. 

We do not know of any computational procedure that takes a description of a
binary linear system $Mx=b$ as input and determines whether or not the game has
a perfect entangled strategy. For tensor-product strategies, the
characterization of perfect strategies in \cite{CleveM2012} can be used to
certify the existence of a perfect tensor-product strategy, but cannot certify
the non-existence of a perfect strategy. Interestingly, the situation seems to
be reversed for commuting-operator strategies. We discuss this in some concluding
remarks at the end of the paper. 

All the results in this paper generalize to linear systems over $\ZZ_p$, $p$ a
prime. For simplicity, we concentrate on the case of binary linear systems
throughout. The generalization to arbitrary primes $p$ is briefly explained in
the concluding remarks as well.

\section{Main results}

We now make some of the definitions from the introduction precise, starting
with the definition of a linear system game. 
\begin{defn}\label{def:nonlocal}
    Let $Mx=b$ be a binary linear system, so $M \in \ZZ_{2}^{m \times
    n}$ and $b \in \ZZ_{2}^m$. In the associated linear system game, Alice
    receives as input $s \in \{1,\ldots,m\}$, and Bob receives $t \in
    \{1,\ldots,n\}$, where $M_{s,t} =1$. Alice outputs an assignment to the
    variables in equation $s$, and Bob outputs a bit. Alice and Bob win
    if Alice's assignment satisfies equation $s$ and Alice's assignment to
    variable $x_t$ is the same as Bob's output bit.
\end{defn}
We postpone the definition of commuting-operator strategies for linear system
games to the following section. The next step is to define the solution group.
\begin{defn}[Solution group of a binary linear system]\label{def:solutiongroup}
    The \emph{solution group} of a binary linear system $Mx=b$ is the group
    $\Gamma$ generated by $g_1, \dots, g_n$ and $J$ satisfying the following
    relations (where $\gi$ is the group identity, and $[a,b] = a b a^{-1}
    b^{-1}$ is the group commutator):
    \begin{enumerate}[(a)]
        \item $g_i^2 = \gi$ for all $1\leq i \leq n$, and $J^2 = \gi$ (generators are involutions).
        \item $[g_i,J] = \gi$ for all $1 \leq i \leq n$ ($J$ commutes with each generator).
        \item If $x_i$ and $x_j$ appear in the same equation (i.e., $i, j \in
            V_{\ell}$ for some $\ell$) then $[g_i,g_j] = \gi$ (local compatibility).
        \item $g_1^{M_{\ell 1}} g_2^{M_{\ell 2}} \cdots g_{n}^{M_{\ell n}} = J^{b_{\ell}}$
            for all $1 \leq \ell \leq m$ (constraint satisfaction).
    \end{enumerate}
\end{defn}
As in the introduction, the last relation can be written as 
\begin{equation*}
    \prod_{i \in V_{\ell}} g_i = g_{k_1} \cdots g_{k_r},
\end{equation*}
where $V_{\ell} = \{k_1,\ldots,k_r\}$ are the indices of variables in equation
$\ell$.

We can now state our main result:
\begin{thm}\label{thm:main}
    Let $Mx=b$ be a binary linear system. Then the following statements are equivalent:
    \begin{enumerate}
        \item There is a perfect commuting-operator strategy for the non-local game associated to $Mx=b$.
        \item There is an operator solution for $Mx=b$ (possibly on an infinite-dimensional Hilbert space).
        \item The solution group for $Mx=b$ has the property that $J \neq \gi$.
    \end{enumerate}
\end{thm}
The proof of Theorem \ref{thm:main} is given in the next section. For comparison,
we note that the main result of \cite{CleveM2012} can also be phrased using the
solution group.
\begin{thm}[\cite{CleveM2012}]\label{thm:compare}
    Let $Mx=b$ be a binary linear system. Then the following statements are equivalent:
    \begin{enumerate}
        \item There is a perfect tensor-product strategy for the non-local game associated to $Mx=b$.
        \item There is a finite-dimensional operator solution for $Mx=b$. 
        \item The solution group for $Mx=b$ has a finite-dimensional representation for which $J \neq \gi$.
    \end{enumerate}
\end{thm}
Although the solution group is not mentioned explicitly in \cite{CleveM2012},
the equivalence with condition (3) is straightforward. The requirement in
\cite{CleveM2012} that the Hilbert spaces $\HH_A$ and $\HH_B$ be separable can
also be dropped, since every entangled state in $\HH_A \otimes \HH_B$ can be
written as 
\begin{equation*}
    \sum_{k=1}^{\infty} \alpha_k \ket{u_k} \otimes \ket{v_k}
\end{equation*}
for some orthonormal sets $\{ \ket{u_k} : k \in \NN\} \subset \HH_A$ and
$\{\ket{v_k} : k \in \NN\} \subset \HH_B$. We thank Vern Paulsen for pointing
this out. 

\section{Proofs}

To prove Theorem \ref{thm:main}, we start by looking at commuting-operator
strategies for linear system games. It is straight-forward (see for instance
\cite{CleveM2012}) that Alice's and Bob's measurements in such a strategy can
be represented by families of binary observables
\begin{equation*}
    \{A_i^{(\ell)} : 1 \leq \ell \leq m, i \in V_{\ell}\} \text{ and } \{B_i : 1 \leq i \leq n \}
\end{equation*}
respectively, where $A_i^{\ell}$ is the observable for Alice's assignment to
variable $x_i$ in equation $\ell$, and $B_j$ is the observable for Bob's
assignment to variable $x_j$. Thus we can formally define commuting-operator
strategies as follows:
\begin{defn}\label{def:strat}
    Let $Mx=b$ be an $m\times n$ binary linear system. A \emph{commuting
    operator strategy} for (the game associated to) $Mx=b$ consists of a Hilbert
    space $\mathcal{H}$, a state $\ket\psi\in \HH$, and two collections
    $\{A_{i}^{(\ell)} : 1 \leq \ell \leq m, i\in V_\ell \}$ and
    $\{B_j : 1 \leq j \leq n\}$ of self-adjoint operators on $\HH$ such that:
    \begin{enumerate}[(a)]
        \item $\bigl(A_i^{(\ell)}\bigr)^2 = B_j^2 = \I$ for all 
            $1 \leq \ell \leq m$, $i \in V_{\ell}$, and $1 \leq j \leq n$.
            ($A_i^{(\ell)}$ and $B_j$ are binary observables).
        \item $A_i^{\ell} B_j = B_j A_i^{(\ell)}$ for all 
            $1 \leq \ell \leq m$, $i \in V_{\ell}$, and $1 \leq j \leq n$.
            (Alice's operators commute with Bob's operators).
        \item $A_{i}^{(\ell)}A_{j}^{(\ell)} = A_{j}^{(\ell)}A_{i}^{(\ell)}$
            for all $1 \leq \ell \leq m$ and $i,j \in V_{\ell}$. 
            (local compatibility).
    \end{enumerate}
\end{defn}
The local compatibility requirement comes from the fact that Alice must measure
the observables $A^{(\ell)}_i$, $i \in V_{\ell}$, simultaneously. Using this
definition, we can identify perfect strategies as follows:
\begin{prop}\label{prop:perfectstrat}
    A commuting-operator strategy $\left( \HH, \ket{\psi}, \bigl\{A_i^{(\ell)}\bigr\}, \bigl\{B_j\bigr\} \right)$
    is perfect if and only if
    \begin{enumerate}[(1)]
        \item $A_i^{(\ell)} \ket\psi = B_i \ket\psi$ for all $1 \leq \ell \leq m$
            and $i \in V_{\ell}$ (consistency between Alice and Bob), and
        \item $\prod_{i\in V_\ell} A_{i}^{(\ell)}\ket\psi = (-1)^{b_\ell}\ket\psi$ 
            for all $1 \leq \ell \leq m$ (constraint satisfaction).
    \end{enumerate}
\end{prop}
\begin{proof}
    Alice's output is always consistent with Bob's if and only if
    \begin{equation*}
        \bra\psi A_i^{(\ell)} B_i \ket\psi = 1
    \end{equation*}
    for all $1 \leq \ell \leq m$ and $i \in V_{\ell}$. But $A_i^{(\ell)} B_i$
    is the product of two unitary operators, and hence is unitary. Since $\ket\psi$
    is a unit vector, the above equation holds if and only if
    \begin{equation*}
        A_i^{(\ell)} B_i \ket\psi = \ket \psi.
    \end{equation*}
    Since $A_i^{(\ell)}$ is an involution, this equation is equivalent to the identity
    in part~(1) of the proposition.

    Similarly, Alice's assignment for equation $\ell$ is always a satisfying assignment if and only if
    \begin{equation*}
        \bra\psi (-1)^{b_\ell} \prod_{i\in V_{\ell}} A_i^{(\ell)} \ket \psi = 1.
    \end{equation*}
    Again, $(-1)^{b_\ell} \prod_{i\in V_{\ell}} A_i^{(\ell)}$ is unitary, so
    the above equation is equivalent to the identity in part (2) of the
    proposition.
\end{proof}

Using Proposition \ref{prop:perfectstrat}, we can prove the first part of
Theorem \ref{thm:main}.
\begin{lemma}\label{lem:oneimpliestwo}
    Let $\left( \HH, \ket{\psi}, \bigl\{A_i^{(\ell)}\bigr\},
    \bigl\{B_j\bigr\} \right)$ be a perfect commuting-operator strategy for
    $Mx=b$, and let $\HH_0 = \overline{\mcA \ket{\psi}}$,
    where $\mcA$ is the unital algebra generated by $\{A_i^{(\ell)}\}$, and
    $\mcA \ket{\psi} = \{A \ket{\psi} : A \in \mcA\}$. Finally, let
    $Q_i := A_i^{(\ell)}|_{\HH_0}$ for some $\ell$ with $i \in V_{\ell}$.  Then
    $Q_1,\ldots,Q_n$ is an operator solution for $Mx=b$.
\end{lemma}
\begin{proof}
    Let $\mcB$ be the unital algebra generated by $\{B_j\}$. By Proposition
    \ref{prop:perfectstrat}, we know that $A_i^{(\ell)} \ket\psi = B_i
    \ket\psi$ for all $i \in V_{\ell}$.  Since $\mcA$ and $\mcB$ commute, if follows immediately that
    for every $A \in \mcA$, there is $B \in \mcB$ such that $A \ket\psi =
    B\ket\psi$. In particular, this tells us that $\mcA \ket\psi = \mcB
    \ket\psi$, and consequently that $\HH_0 = \overline{\mcB \ket \psi}$. 

    Now suppose we have $A, A' \in \mcA$ such that $A\ket\psi = A'\ket\psi$.
    Then 
    \begin{equation*}
        A B \ket\psi = B A \ket\psi = B A'\ket\psi = A' B\ket\psi
    \end{equation*}
    for all $B \in \mcB$. By continuity, we conclude that $A|_{\HH_0} =
    A'|_{\HH_0}$. Suppose that variable $x_i$ belongs to equations $\ell$ and
    $\ell'$, or in other words that $i \in V_{\ell} \cap V_{\ell'}$. Then
    \begin{equation*}
        A_i^{(\ell)}\ket\psi = B_i \ket\psi = A_i^{(\ell')} \ket\psi
    \end{equation*}
    by Proposition \ref{prop:perfectstrat} again. We conclude that
    $A_i^{(\ell)}|_{\HH_0} = A_i^{(\ell')}|_{\HH_0}$, and thus $Q_i =
    A_i^{(\ell)}|_{\HH_0}$ is independent of the choice of $\ell$.

    We can now check that $Q_1,\ldots,Q_n$ is an operator solution. Since
    $\HH_0$ is $\mcA$-invariant,
    \begin{equation*}
        Q_i^2 = \bigl(A_i^{(\ell)}\bigr)^2\,|_{\HH_0} = \I_{\HH_0}.
    \end{equation*}
    Similarly, if $i$ and $j$ both belong to $V_{\ell}$, then
    \begin{equation*}
        Q_i Q_j = A_i^{(\ell)} A_j^{(\ell)} |_{\HH_0}
        = A_j^{(\ell)} A_i^{(\ell)}|_{\HH_0} = Q_j Q_i.
    \end{equation*}
    Finally, 
    \begin{equation*}
        \prod_{i \in V_{\ell}} A_i^{(\ell)} \ket\psi = (-1)^{b_{\ell}} \ket\psi
    \end{equation*}
    by Proposition \ref{prop:perfectstrat}, and hence
    \begin{equation*}
        \prod_{i \in V_{\ell}} Q_i = \prod_{i \in V_{\ell}} A_i^{(\ell)} |_{\HH_0}
            = (-1)^{b_{\ell}} \I_{\HH_0}
    \end{equation*}
    for all $1 \leq i \leq \ell$.  
\end{proof}

The second part of Theorem \ref{thm:main} is easy to prove.
\begin{lemma}\label{lem:twoimpliesthree}
    If $Mx=b$ has an operator solution then $J \neq \gi$ in the solution
    group $\Gamma$ of $Mx=b$.
\end{lemma}
\begin{proof}
    Suppose $A_1,\ldots,A_n$ is an operator solution for $Mx=b$. 
    By Definitions \ref{def:operator-solution} and \ref{def:solutiongroup}, 
    the map sending
    \begin{equation*}
        g_i \mapsto A_i,\; 1 \leq i \leq n \quad \text{ and } \quad J \mapsto -\I
    \end{equation*}
    is a representation of $\Gamma$ with $J \neq \I$. It follows that $J \neq
    \gi$ in $\Gamma$. 
\end{proof}

\begin{proof}[Proof of Theorem \ref{thm:main}]
    We have shown in Lemmas \ref{lem:oneimpliestwo} and
    \ref{lem:twoimpliesthree} that (1) implies (2) and (2) implies (3). It
    remains to show that (3) implies (1). Suppose that $J \neq \gi$ in the
    solution group $\Gamma$. We need to construct a perfect commuting-operator
    strategy for $Mx=b$. Let
    \begin{equation*}
        \HH = \left\{\sum_{g\in \Gamma} \alpha_g \ket g : \alpha_g \in \CC \text{ such that } \sum_{g \in \Gamma} 
            |\alpha_g|^2 < \infty\right\}
    \end{equation*}
    be the completion of the group algebra of $\Gamma$. Given $g \in \Gamma$,
    let $L_g$ and $R_g$ be the left and right multiplication operators for $g$
    on $\HH$, so 
    \begin{equation*}
        L_g \ket h = \ket{g h}\quad \text{ and }\quad R_g \ket h = \ket{hg}.
    \end{equation*}
    Clearly, $L_g$ and $R_g$ are unitary. Furthermore, 
    \begin{equation*}
        L_g R_h = R_h L_g, \quad L_g L_h = L_{gh}, \quad \text{ and } \quad R_g R_h = R_{hg}
    \end{equation*}
    for all $h,g \in \Gamma$. We set
    \begin{equation*}
        A_i^{(\ell)} := L_{g_i}\; \text{ for all } 1 \leq \ell \leq m,\; i \in V_{\ell},
    \end{equation*}
    and
    \begin{equation*}
        B_j := R_{g_i}\; \text{ for all } 1 \leq j \leq n.
    \end{equation*}
    Finally we set
    \begin{equation*}
        \ket\psi := \frac{\ket{e} - \ket{J}}{\sqrt{2}}.
    \end{equation*}
    Since $J \neq \gi$ in $\Gamma$, $\ket\psi$ is a well-defined unit vector in
    $\HH$. 
    Since 
    \begin{equation*}
        L_{g_i}^2 = L_{g_i^2} = L_e = \I = R_{g_i}^2 \text{ for all } 1 \leq i \leq n
    \end{equation*}
    and
    \begin{equation*}
       L_{g_i} L_{g_j} = L_{g_i g_j} = L_{g_j g_i} = L_{g_j} L_{g_i} \text{ for all }
            1 \leq \ell \leq m \text{ and } i,j \in V_{\ell},
    \end{equation*}
    it is clear that $\{A_i^{(\ell)}\}$, $\{B_i\}$, and $\ket\psi$ form a valid
    commuting-operator strategy for $Mx=b$. To show that they form a perfect
    strategy, observe that 
    \begin{equation*}
        A_i^{(\ell)} \ket\psi = \frac {\ket{g_i} - \ket{g_i J}}{\sqrt 2}  = \frac {\ket{g_i} - \ket{J g_i}}{\sqrt 2} =          B_i \ket\psi
    \end{equation*}
    for all $1 \leq \ell \leq m$ and $i \in V_{\ell}$, and that
    \begin{equation*}
        \prod_{i\in V_\ell} A_i\ket\psi= \prod_{i\in V_\ell} L_{g_i}\ket\psi= 
L_{\!J^{{\raisebox{1.6pt}{$\scriptscriptstyle b_{\ell}$}}}} \ket\psi, 
    \end{equation*}
    for all $1 \leq \ell \leq m$. If $b_\ell = 0$, then 
    \begin{equation*}
        L_{\!J^{{\raisebox{1.6pt}{$\scriptscriptstyle b_{\ell}$}}}} \ket\psi = L_e \ket\psi = \ket\psi,
    \end{equation*}
    while if $b_\ell = 1$ then
    \begin{equation*}
        L_{\!J^{{\raisebox{1.6pt}{$\scriptscriptstyle b_{\ell}$}}}} \ket\psi = L_{\raisebox{-1.6pt}{$\scriptstyle J$}} \ket\psi = \frac {\ket{J} -\ket{e}}{\sqrt 2} =-\ket\psi.
    \end{equation*}
Therefore, $\prod_{i\in V_\ell} A_i\ket\psi = (-1)^{b_\ell} \ket\psi$ for all 
    $1 \leq \ell \leq m$. Thus the strategy we have constructed is perfect by Proposition
    \ref{prop:perfectstrat}.
\end{proof}

\section{Concluding remarks}

As mentioned in the introduction, we do not know of any computational procedure
which can determine if a binary linear system has a perfect entangled strategy. 
Arkhipov showed that, in the special case where each variable appears in
exactly two constraints, there is a polynomial-time algorithm to determine if a
perfect entangled strategy exists \cite{Arkhipov2012} (in this case, a game has
a perfect commuting-operator strategy if and only if it has a perfect
tensor-product strategy). For the general case, we can attempt to use
the characterization 
of perfect strategies in~\cite{CleveM2012} by searching for operator
solutions over $\CC^d$, $d \in \NN$. It is decidable to determine if there is
an operator solution over $\CC^d$ for fixed $d$, and thus this naive procedure
is guaranteed to find a perfect strategy if one exists. However, if a perfect
strategy does not exist, then the naive procedure does not halt. We note that, for arbitrarily large $d$, Ji gives examples of binary linear systems which have finite-dimensional
operator solutions, but for which the solutions require dimension at least $d$~\cite{Ji2013}.

In contrast, there is no apparent way to search through operator solutions over
infinite-dimensional Hilbert spaces. What we can do instead is try to show that
$J=\gi$ in the group $\Gamma$ by searching through products of the defining
relations. Using our characterization, we see that this procedure will halt if
and only if the linear system game does not have a perfect strategy in the
commuting-operator model. Thus this problem would be decidable if the
tensor-product model and commuting-operator model were equivalent. Determining
whether or not these two models are equivalent is a well-known open problem due
to Tsirelson \cite{Tsirelson1993}. 

As also mentioned in the introduction, the results in this paper generalize to
linear systems over $\ZZ_p$. The non-local game associated to a system over
$\ZZ_p$ is defined in exactly the same way, although Alice and Bob output
assignments from $\ZZ_p$ rather than $\ZZ_2$. Similarly, commuting-operator
strategies are modelled using 
measurements based on unitary operations
$U$ with $U^p = \I$, rather
than $U^2=\I$. Likewise, the definition of the solution group must be changed
so that $g_i^p=\gi$ and $J^p=\gi$. Finally, the state $\ket{\psi}$ in the
proof of Theorem \ref{thm:main} becomes
\begin{equation*}
    \frac{1}{\sqrt{p}} \sum_{i=0}^{p-1} \zeta^{-i} \ket{J^i},
\end{equation*}
where $\zeta$ is a primitive $p$th root of unity. Otherwise all definitions,
propositions, and proofs are the same. 

\section{Acknowledgements}
We would like to thank Vern Paulsen and Zhengfeng Ji for many helpful discussions. This research was supported in part by Canada's NSERC, a David R. Cheriton Scholarship, and a Mike and Ophelia Lazaridis Fellowship.

\bibliographystyle{amsplain}

\bibliography{entanglement}

\providecommand{\bysame}{\leavevmode\hbox to3em{\hrulefill}\thinspace}
\providecommand{\MR}{\relax\ifhmode\unskip\space\fi MR }
\providecommand{\MRhref}[2]{%
  \href{http://www.ams.org/mathscinet-getitem?mr=#1}{#2}
}
\providecommand{\href}[2]{#2}
\begin{thebibliography}{10}

\bibitem{Aravand2004}
P.~K. Aravind, \emph{Quantum mysteries revisited again}, American Journal of
  Physics \textbf{72} (2004), 1303--1307.

\bibitem{Arkhipov2012}
A.~Arkhipov, \emph{Extending and characterizing quantum magic games},
  arXiv:1209.3819 (2012).

\bibitem{CleveM2012}
R.~Cleve and R.~Mittal, \emph{Characterization of binary constraint system
  games}, Proceedings of the 41st International Colloquium on Automata,
  Languages, and Programming (ICALP), 2012, pp.~320--331.

\bibitem{CleveH+2004}
R.~Cleve, P.~H\o yer, B.~Toner, and J.~Watrous., \emph{Consequences and limits
  of nonlocal strategies}, Proceedings of the 19th IEEE Conference on
  Computational Complexity (CCC), 2004, pp.~236--249.

\bibitem{Fritz2012}
T.~Fritz, \emph{Tsirelson's problem and {K}irchberg's conjecture}, Reviews in
  Mathematical Physics \textbf{24} (2012), no.~5, 1250012.

\bibitem{Ji2013}
Z.~Ji, \emph{Binary constraint system games and locally commutative
  reductions}, arXiv:1310.3794 (2013).

\bibitem{junge2011}
Marius Junge, Miguel Navascues, Carlos Palazuelos, D~Perez-Garcia, Volkher~B
  Scholz, and Reinhard~F Werner, \emph{Connes' embedding problem and
  {T}sirelson's problem}, Journal of Mathematical Physics \textbf{52} (2011),
  no.~1, 012102.

\bibitem{Mermin1990}
N.~D. Mermin, \emph{Simple unified form for the major no-hidden-variables
  theorems}, Physical Review Letters \textbf{65} (1990), no.~27, 3373--3376.

\bibitem{Mermin1993}
\bysame, \emph{Hidden variables and the two theorems of {J}ohn {B}ell}, Reviews
  of Modern Physics \textbf{65} (1993), no.~3, 803--815.

\bibitem{PaulsenSSTW14}
V.~I. Paulsen, S.~Severini, D.~Stahlke, I.~G. Todorov, and A.~Winter,
  \emph{Estimating quantum chromatic numbers}, 2014, Manuscript available at
  arXiv1407.6918.

\bibitem{Peres1990}
A.~Peres, \emph{Incompatible results of quantum measurements}, Physics Letters
  A \textbf{151} (1990), no.~3,4, 107--108.

\bibitem{scholz2008}
Volkher~B. Scholz and Reinhard~F. Werner, \emph{Tsirelson's problem}, arXiv
  preprint arXiv:0812.4305 (2008).

\bibitem{Tsirelson1993}
B.~S. Tsirelson, \emph{Some results and problems on quantum {B}ell-type
  inequalities}, Hadronic Journal Supplement \textbf{8} (1993), 329--345.

\end{thebibliography}

\end{document}